\documentclass[journal, a4paper]{IEEEtran}

\makeatletter

\let\proof\@undefined
\let\endproof\@undefined
\makeatother

\usepackage{report}

\newcommand{\eps}{\varepsilon}
\def\Tset{\set{A}^{[n]}(p, \delta)}
\def\KLTset{\set{A}^{[n]}(p, q, \delta)}
\newcommand{\KLTsetd}[1]{\set{A}^{[n]}(p, q, #1)}
\def\pn{p^{[n]\hspace{-0.5pt}}}
\def\pnp{p^{[n]'\hspace{-0.5pt}}}
\def\pone{p^{[1]\hspace{-0.5pt}}}
\def\qn{q^{[n]\hspace{-0.5pt}}}
\def\qnp{q^{[n]'\hspace{-0.5pt}}}
\def\qone{q^{[1]\hspace{-0.5pt}}}
\def\Dn{D^{[n]\hspace{-0.5pt}}}
\def\Done{D^{[1]\hspace{-0.5pt}}}
\def\hn{h^{[n]\hspace{-0.5pt}}}
\def\hone{h^{[1]\hspace{-0.5pt}}}
\def\Hn{H^{[n]\hspace{-0.5pt}}}
\def\Hone{H^{[1]\hspace{-0.5pt}}}
\def\deltan{\delta^{[n]\hspace{-0.5pt}}}
\def\deltaone{\delta^{[1]\hspace{-0.5pt}}}
\def\epsilonn{\epsilon^{[n]\hspace{-0.5pt}}}
\def\epsilonone{\epsilon^{[1]\hspace{-0.5pt}}}
\def\epsn{\eps^{[n]\hspace{-0.5pt}}}

\def\taun{\tau^{[n]\hspace{-0.5pt}}}
\def\tauone{\tau^{[1]\hspace{-0.5pt}}}
\def\gamman{\gamma^{[n]\hspace{-0.5pt}}}

\def\An{\mat{A}^{[n]\hspace{-0.5pt}}}
\def\Bn{\mat{B}^{[n]\hspace{-0.5pt}}}
\def\Dn{\mat{D}^{[n]\hspace{-0.5pt}}}
\def\Cn{\mat{C}^{[n]\hspace{-0.5pt}}}
\def\Ani{\mat{A}^{-1 [n]\hspace{-0.5pt}}}
\def\Bni{\mat{B}^{-1 [n]\hspace{-0.5pt}}}

\def\Ln{\Lambda^{[n]\hspace{-0.5pt}}}
\def\Lnh{\hat{\Lambda}^{[n]\hspace{-0.5pt}}}
\def\Lnt{\tilde{\Lambda}^{[n]\hspace{-0.5pt}}}
\def\dtf{\! \left( e^{j 2 \pi f} \right)}

\newcommand{\vol}[1]{\mathsf{vol}\!\left(#1\right)}

\usepackage{cite}
\usepackage{scalerel}

\def\bs{\mkern-12mu} 

\begin{document}
\title{The Generalized Chernoff–Stein Lemma, Applications and Examples} 

\author{%
 \IEEEauthorblockN{Ibrahim Abou Faycal, Jihad Fahs, and Ibrahim Issa}
 \IEEEauthorblockA{Department of Electrical and Computer Engineering \\
                   American University of Beirut, Beirut, Lebanon \\
                   Emails: \{Ibrahim.Abou-Faycal, Jihad.Fahs, Ibrahim.Issa\}@aub.edu.lb}
}


\maketitle


%

\section{Introduction}

In this manuscript we define the notion of "$\delta$-typicality" for both entropy and relative entropy, as well as a notion of $\epsilon$-goodness and provide an extension to Stein's lemma for continuous quantities as well as correlated setups. We apply the derived results on the Gaussian hypothesis testing problem where the observations are possibly correlated.

\section{$\delta$-Typical Sets}
\label{sec:Typical}

We start by extending the classical framework to dependent random
variables.
 
Denote by $\pn(\cdot)$ a --family of-- probability distribution on
$\set{X}^n$. We use the label $p$ to refer to $\pn(\cdot)$ for the
dimension $n$ at hand. In simple terms, one may think of $p$ as a
label for the family of laws indexed by $n$
\begin{equation*}
  p: \left\{ \pn(\cdot) \right\}_n .
\end{equation*}

Depending on whether the laws are discrete or absolutely continuous,
we also define
\begin{equation*}
  \left.
    \begin{array}{c}
      \hn \eqdef h \! \left( \pn \right) \\
      \Hn \eqdef H \! \left( \pn \right)
    \end{array}
  \right\} = - \Ep{\pn}{\log \pn(\vX)}.
\end{equation*}

\smallskip

\begin{example*}[IID]
  The example corresponding to an IID scenario is given by
  \begin{equation*}
    \pn(\vx) = \prod_k \pone(x_k),
  \end{equation*}
  and in that scenario,
  \begin{equation*}
    \hn = \, n h \! \left( \pone \right) = n \hone \quad \text{or} \quad \Hn = \! n \Hone,
  \end{equation*}
  which is obtained through simple manipulations.
\end{example*}

For simplicity, we will use the differential entropy notation in what
follows albeit the statements will hold almost verbatim for the
discrete/entropy case. Whenever the statements need to be tailored to
the discrete case, we will explicitely do so.

Denote by $\delta$ a family $\left\{ \deltan \right\}_n$ of positive
scalars that vary with $n$ at a rate to be specified later.

\begin{definition}
  We define the {\bf $\delta$-typical\/} set $\Tset$
  \begin{multline*}
    \Tset = \\
    \left\{ \vx \in \set{X}^n \!: \hn - \deltan \leq - \log \pn(\vx) \leq \hn + \deltan \right\}.
  \end{multline*}
\end{definition}

Naturally, since we are interested in setups where $n$ is varying and
increasing to infinity, one can talk about the family of
$\delta$-typical sets.

\begin{example*}[IID]
  In the IID scenario,
  \begin{align*}
    & \Tset \\
    = & \hspace{-1pt}\left\{ \hspace{-1pt}\vx \in \set{X}^n \!: n \hone - \deltan \hspace{-1pt} \leq - \hspace{-1pt}\sum_k \log \pone(x_k) \leq n \hone \hspace{-0.5pt} + \hspace{-0.5pt}\deltan \right\} \\
    = &  \hspace{-1pt} \left\{ \hspace{-1pt}\vx \in \set{X}^n \!: \hone - \frac{\deltan}{n} \hspace{-1pt}\leq - \frac{1}{n} \hspace{-1pt}\sum_k \log \pone(x_k) \leq \hone \hspace{-0.5pt} + \hspace{-0.5pt}\frac{\deltan}{n} \hspace{-1pt}\right\}
  \end{align*}
  which is the {\em typical set\/} defined and used by
  Cover~\cite{cover} with $\epsilon = \deltan / n$.
\end{example*}

Let $\epsilon$ be a family $\left\{ \epsilonn \right\}_n$ of positive
scalars that possibly vary with $n$. Mostly, the family $\epsilon$
will be non-increasing with $n$ at a rate of interest to be considered
later.
\begin{definition}
  \label{def:good}
  For a given family $\epsilon$, $\delta$ is said to be
  {\bf $\epsilon$-good\/} if
  \begin{equation}
    \pn \!\left( \Tset \right) \geq 1 - \epsilonn,
    \label{eq:deltaprop}
  \end{equation}
  for $n$ large enough.
\end{definition}

In layman terms, whenever $\delta$ is
$\epsilon$-good the probability --under $\pn$-- of the $\delta$-typical set is within $\epsilonn$ to one. Note that by taking large (see maximal) and small (see minimal) values, it is clear that $\epsilon$-good families always exist for any
$\epsilon$. However, ones with tighter bounds would be naturally
preferred. We consider hereafter $\epsilon$-good $\delta$ families and
attempt to find as tight bounds as possible.

\begin{example}[IID]
  In that scenario, since $\{ X_k \}_k$'s are IID according to
  $\pone$, by the Weak Law of Large Numbers (WLLN),
  \begin{align*}
    - \frac{1}{n} \sum_{k=1}^n \log \pone(X_k) \xrightarrow[n\to\infty]{p} \hone.
  \end{align*}
  Therefore, for the $\epsilon$-family $\epsilonn = \epsilonone$, for
  any $\xi > 0$, the family
  $\left\{ \deltan = n \, \deltaone = n \xi \right\}$ is
  $\epsilon$-good.

  One can identify possibly tighter good families: Denote by $\{Y_k\}$
  the random variables $Y_k = - \log \pone(X_k)$ which are IID
  whenever the $\{X_k\}$'s are. Assuming that the mean and variance of
  $Y_k$ exist, by the Central Limit Theorem (CLT)
  \begin{equation*}
    \frac{1}{\sqrt{n}} \sum_{k=1}^n \Bigl( Y_k - \hone \Bigr) \xrightarrow[n\to\infty]{d}
    \Normal{0}{\sigma^2},
  \end{equation*}
  where $\sigma^2 = \V{- \log \pone(X_1)}$. Therefore, since
  \begin{align*}
    & \, \pn \left( \Tset \right) \\
    & \, = \, \pn \left( - \deltan \leq \sum_{k=1}^n \Bigl( Y_k - \hone \Bigr) \leq \deltan \right) \\
    & \, = \, \pn \left( - \frac{\deltan}{\sqrt{n}} \leq  \frac{1}{\sqrt{n}} \sum_{k=1}^n \Bigl( Y_k
	   - \hone \Bigr) \leq \frac{\deltan}{\sqrt{n}} \right) \\
    & \, = \, 1 - 2 \, \Qfct{ \frac{\deltan}{\sigma \sqrt{n} } }
  \end{align*}
  and $\delta$ is $\epsilon$-good for families such that
  \begin{equation*}
    \epsilonn \geq 2 \, \Qfct{ \frac{\deltan}{\sigma \sqrt{n}} \, } \qquad \text{for } n \text{ large enough}.
  \end{equation*}
  Note that whenever the inequality is not asymptotically satisfied,
  $\delta$ will not be $\epsilon$-good.

  Since $\iQfct{\cdot\,}$ is decreasing, for any given $\epsilon$ the
  family $\delta$ is $\epsilon$-good if and only if
  \begin{equation}
    \deltan \geq \sigma \sqrt{n} \, \iQfct{ \frac{\epsilonn}{2} \, } \qquad \text{for } n \text{ large enough}.
    \label{eq:IIDeG}
  \end{equation}
  \label{ex:TTDG}
\end{example}

Through its definition, one can readily derive the following
properties of $\delta$-typical sets. Note that in the discrete case
the "volume" of the $\delta$-typical set refers to its cardinality.
\begin{property}
  \label{prop:Tset_1}
  The set $\Tset$ satisfies the following properties:
  \begin{itemize}
  \item[(a)] For $\vx \in \Tset$,
    \begin{align*}
      2^{- \left( \hn + \deltan \right)} & \leq \pn(\vx) \leq  2^{- \left( \hn - \deltan \right)}.
    \end{align*}
    
  \item[(b)] $\vol{\Tset} \leq 2^{\left( \hn + \deltan \right)}$.

  \item[(c)] For an $\epsilon$-good family $\delta$, whenever $n$ is large enough,
    \begin{equation*}
      \left( 1 - \epsilonn \right) \, 2^{\left( \hn - \deltan \right)} \leq \vol{\Tset}.
    \end{equation*}
  \end{itemize}
\end{property}

These properties parallel the "classical" ones (found in Cover~\cite{cover} for example) where the interpretation is that of "equipartition" on the typical set.

In the following lemma we consider another "set" that has a probability within
$\epsn$ to one under $\pn$. We show that its volume is "roughly" the
same.

\begin{lemma}
  \label{lm:OtherSets}
  Let $B^{[n]} \subset \set{X}^n$ be any subset of sequences
  $\vx \in \set{X}^n$ such that
  $\pn \! \left( B^{[n]} \right) \geq \left(1 - \epsn\right)$ when $n$ is
  large enough, and where $\epsn$ is a family of positive scalars less than one. For an $\epsilon$-good $\delta$,
  \begin{equation*}
    \vol{ B^{[n]} } \geq \left( 1 - \epsn - \epsilonn \right) \, 2^{\left( \hn - \deltan \right)},
  \end{equation*}
  whenever $n$ is large enough.
\end{lemma}

\begin{proof}
  Consider $n$ to be large enough for the inequalities assumed in the
  statement and in Property~\ref{prop:Tset_1} to hold.

  Since $\pn \! \left( B^{[n]} \right) \geq \left(1 - \epsn\right)$ and
  $\pn \! \left( A^{[n]} \right) \geq \left( 1 - \epsilonn \right)$, by
  the union bound
  \begin{equation*}
    \pn \! \left( A^{[n]} \cap B^{[n]} \right) \geq \left( 1 - \epsilonn - \epsn \right).
  \end{equation*}
  
  On the other hand, 
  \begin{align*}
    \pn \! \left( A^{[n]} \cap B^{[n]} \right) & = \int_{\vx \in A^{[n]} \cap B^{[n]} } \pn (\vx) \, d\vx \\
		      & \leq \int_{\vx \in A^{[n]} \cap B^{[n]} } 2^{- \left( \hn - \deltan \right)} \, d\vx \\
		      & = \, 2^{- \left( \hn - \deltan \right)} \, \vol{ A^{[n]} \cap B^{[n]} } \\
		      & \leq \, 2^{- \left( \hn - \deltan \right)} \, \vol{ B^{[n]} } \\
    \Rightarrow \quad \vol{ B^{[n]} } & \geq \left( 1 - \epsilonn - \epsn \right) \, 2^{\left( \hn - \deltan \right)}.
  \end{align*}

  Naturally, the derivations hold for the discrete case with
  summations replacing integrals and cardinality replacing volume.
\end{proof}

\subsection{The Dependent Gaussian Case}
\label{sec:typicalG}

We determine the $\delta$-typical set and some $\epsilon$-good
$\delta$ families for the {\em dependent\/} Gaussian case
$\vX \sim \Normal{\vzero}{\Lambda}$.

Using the expression of the Probability Density Function (PDF) and
evaluating in nats,
\begin{align*}
  \pn(\vx) & = \frac{1}{\sqrt{(2 \pi)^n  \text{det} \left( \Lambda \right)}} \, e^{-\frac{1}{2} \vxt \Lambda^{-1} \vx} \\
  h(\vX) & = \ln \sqrt{ \left( 2 \pi e \right)^n \text{det} \left( \Lambda \right) } \\
    & = \ln \sqrt{ \left( 2 \pi \right)^n \text{det} \left( \Lambda \right) } + \frac{n}{2} \\
  - \ln \pn(\vx) - h(\vX) & = \frac{1}{2} \, \vxt \Lambda^{-1} \vx - \frac{n}{2} \\
  \Tset & =
	  \left\{ \vx \in \set{X}^n \!: \left| \vxt \Lambda^{-1} \vx
	  - n \right| \leq 2 \, \deltan \right\},
\end{align*}
which is the set of vectors the norm of which is within $2 \, \deltan$
of $n$, where by norm we mean the one induced by the positive definite
matrix $\Lambda^{-1}$.

Examining the probability of the $\delta$-typical set,
\begin{align*}
  \pn \left( \Tset \right) & = p_{\vX} \sbs \left( n - 2 \, \deltan \leq \vXt \Lambda^{-1} \vX
			     \leq n + 2 \, \deltan \right) \\
  & = p_{\vY} \sbs \left( n - 2 \, \deltan \leq \| \vY \|^2 \leq n + 2 \, \deltan \right),
\end{align*}
where $\vY = \Lambda^{-1/2} \vX \sim \Normal{\vzero}{\mat{I}}$ and
where $\Lambda^{-1/2}$ is defined as follows: Diagonalize $\Lambda$ using orthogonal matrix $\mat{U}$ and denote the diagonal matrix by $\Delta$. Denoting by $\Delta^{1/2}$ the diagonal matrix of the square-roots of the eigenvalues and by $\Delta^{-1/2}$ its inverse,
\begin{equation} \begin{array}{ccc}
  \Lambda = \mat{U}^t \Delta \mat{U} \quad & \Rightarrow  & \Lambda^{-1} \! = \mat{U}^t \Delta^{-1} \mat{U} \\
  \Lambda^{1/2} = \mat{U}^t \Delta^{1/2} \mat{U} \quad & \Rightarrow & \Lambda^{-1/2} = \mat{U}^t \Delta^{-1/2} \mat{U}.
    \end{array}
    \label{eq:matSQRT}
\end{equation}

The weak law of large numbers establishes that for any $\xi > 0$, the
family $\left\{ \deltan = n \, \deltaone = n \, \xi \right\}$ is
$\epsilon$-good when the $\epsilonn$'s are constant.

One can determine tighter $\epsilon$-good families by using the CLT as
in Example~\ref{ex:TTDG}:
\begin{equation*}
  \frac{1}{\sqrt{n}} \sum_{k=1}^n \Bigl( Y_k^2 - 1 \Bigr) \xrightarrow[n\to\infty]{d}
  \Normal{0}{2},
\end{equation*}
and since
\begin{align*}
  \pn \left( \Tset \right) & = p_{\vY} \sbs \left( - 2 \, \deltan \leq \sum \left( Y_k^2 - 1 \right) \leq 2
			     \, \deltan \right),
\end{align*}
the following property is readily established.

\begin{property}
  The family $\delta$ is $\epsilon$-good if and only if
  \begin{equation*}
    \epsilonn \geq 2 \, \Qfct{ \frac{\sqrt{2} \, \deltan}{ \sqrt{n}} \, } \qquad \text{for } n \text{ large enough}.
  \end{equation*}
  
  Equivalently: given $\epsilon$, the family $\delta$ is
  $\epsilon$-good if and only if
  \begin{equation*}
    \deltan \geq \sqrt{\frac{n}{2}} \, \iQfct{ \frac{\epsilonn}{2} \, } \qquad \text{for } n \text{ large enough}.
  \end{equation*}
  \label{prop:epsdeltaCG}
\end{property}

Based on the above, one can infer the following:
\begin{property}
  If one desires to consider $\delta$ families that are
  $\epsilon$-good for all constants $\epsilonn$, then
  \begin{equation*}
    \frac{\deltan}{\sqrt{n}} \xrightarrow[n\to\infty]{} \infty
  \end{equation*}
  is $\epsilon$-good for all constants. Conversly, any family where
  $\deltan/\sqrt{n}$ does not grow to infinity is not $\epsilon$-good
  for all constants $\epsilonn$.
\end{property}


\section{Relative Entropy $\delta$-Typical Sets}
\label{sec:KLtypical}

Denote by $\pn(\cdot)$ and $\qn(\cdot)$ two --families of--
probability distributions on $\set{X}^n$ and let $p$ and $q$ be labels for the
two respective families of laws indexed by $n$,
\begin{equation*}
  p: \left\{ \pn(\cdot) \right\}_n , \qquad q: \left\{ \qn(\cdot) \right\}_n.
\end{equation*}

We also define
\begin{equation*}
  \Dn \eqdef \KL{\pn}{\qn} = \Ep{\pn}{\log \frac{\pn(\vX)}{\qn(\vX)}}.
\end{equation*}

We consider both scenarios where, the laws $\left\{ \pn(\cdot)\right\}_n$
and $\left\{ \qn(\cdot)\right\}_n$ are
\begin{itemize}
\item[(i)] discrete, in which case
  \begin{equation*}
    \Dn = \sum_{\vx \in \set{X}^n} \pn(\vx) \log \frac{\pn(\vx)}{\qn(\vx)}.
  \end{equation*}
\item[(ii)] absolutely continuous, in which case
  \begin{equation*}
    \Dn = \int_{\vx \in \set{X}^n} \pn(\vx) \log \frac{\pn(\vx)}{\qn(\vx)} \, d\vx.
  \end{equation*}
\end{itemize}

\smallskip

\begin{example*}[IID]
  The example corresponding to an Independent and Identically
  Distributed (IID) scenario is given by
  \begin{equation*}
    \pn(\vx) = \prod_k \pone(x_k), \qquad \qn(\vx) = \prod_k \qone(x_k).
  \end{equation*}

  In that scenario,
  \begin{equation*}
    \Dn = \, n \KL{\pone}{\qone} = n \Done,
  \end{equation*}
  which is obtained through simple manipulations
  \begin{align*}
    \Dn = & \, \Ep{\pn}{\log \frac{\pn(\vX)}{\qn(\vX)}} = \Ep{\pn}{\sum_{k=1}^n \log \frac{\pone(X_k)}{\qone(X_k)}}  \\
    = & \, \sum_{k=1}^n \Ep{\pn}{\log \frac{\pone(X_k)}{\qone(X_k)}} = n \KL{\pone}{\qone}.
  \end{align*}
\end{example*}

Denote by $\delta$ a family $\left\{ \deltan \right\}_n$ of positive
scalars that vary with $n$ at a rate to be specified later.

\begin{definition}
  We define the relative-entropy {\bf $\delta$-typical\/} set $\KLTset$ to be
  \begin{multline*}
    \KLTset = \\
    \left\{ \vx \in \set{X}^n \!: \Dn - \deltan \leq \log
      \frac{\pn(\vx)}{\qn(\vx)} \leq \Dn + \deltan \right\}.
  \end{multline*}
\end{definition}

Naturally, since we are interested in setups where $n$ is varying and increasing to infinity, one can talk about the family of $\delta$-typical sets.

\begin{example*}[IID]
  In the IID scenario,
  \begin{align*}
    & \KLTset \\
     = & \left\{ \vx \in \set{X}^n \!: n \Done - \deltan \leq \sum_k \log
	    \frac{\pone(x_k)}{\qone(x_k)} \leq n \Done + \deltan \right\} \\
    = & \left\{ \vx \in \set{X}^n \!: \Done - \frac{\deltan}{n} \leq \frac{1}{n} \sum_k \log
      \frac{\pone(x_k)}{\qone(x_k)} \leq \Done + \frac{\deltan}{n} \right\}
  \end{align*}
  which is the {\em relative entropy typical set\/} defined and used by Cover~\cite{cover} with $\epsilon = \deltan / n$.
\end{example*}

\begin{definition}
  \label{def:KLgood}
  For a given family $\epsilon$, $\delta$ is said to be
  {\bf $\epsilon$-good\/} if
  \begin{equation}
    \pn \!\left( \KLTset \right) \geq 1 - \epsilonn,
    \label{eq:KLdeltaprop}
  \end{equation}
  for $n$ large enough. 
\end{definition}

Said differently, $\delta$ is $\epsilon$-good if the probability
--under $p$-- of the $\delta$-typical is within $\epsilon$ from one as
$n$ goes to infinity.

\begin{example*}[IID]
  Consider the case where $\{ X_k \}_k$'s are IID according to
  $\pone$, and $\qn(\vx) = \prod \qone(x_k)$. By the WLLN,
  \begin{align*}
    \frac{1}{n} \sum_{k=1}^n \log \frac{\pone(X_k)}{\qone(X_k)} \xrightarrow[n\to\infty]{p} \Done,
  \end{align*}
  and hence, for any $\xi > 0$, the family
  $\left\{ \deltan = n \, \deltaone = n \xi \right\}$ is
  $\epsilon$-good for all $\epsilon$ families that are constant.

  To determine tighter good families as in Example~\ref{ex:TTDG}, denote by $\{Z_k\}$ the random variables $Z_k = \log \left( \pone(X_k) / \qone(X_k) \right)$ which are IID. Note that the mean --under $\pone$-- of $Z_k$'s is the relative entropy $\KL{\pone}{\qone}$. Assuming that this mean exists and that additionally $\sigma^2$ the variance of $Z_k$ exists, by the CLT,
  \begin{equation*}
    \frac{1}{\sqrt{n}} \sum_{k=1}^n \Bigl( Z_k - \Done \Bigr) \xrightarrow[n\to\infty]{d}
    \Normal{0}{\sigma^2}.
  \end{equation*}
  
  Therefore, since
  \begin{align*}
    & \pn \left( \KLTset \right) \\
    & \, = \, \pn \left( - \deltan \leq \sum_{k=1}^n \Bigl( Z_k - \Done \Bigr) \leq \deltan \right) \\
    & \, = \, \pn \left( - \frac{\deltan}{\sqrt{n}} \leq  \frac{1}{\sqrt{n}} \sum_{k=1}^n \Bigl( Z_k
	   - \Done \Bigr) \leq \frac{\deltan}{\sqrt{n}} \right) \\
    & \, = \, 1 - 2 \, \Qfct{ \frac{\deltan}{\sigma \sqrt{n} } }
  \end{align*}
  and $\delta$ is $\epsilon$-good for families such that
  \begin{equation*}
    \epsilonn \geq 2 \, \Qfct{ \frac{\deltan}{\sigma \sqrt{n}} \, } \qquad \text{for } n \text{ large enough}.
  \end{equation*}
  Note that whenever the inequality is not asymptotically satisfied,
  $\delta$ will not be $\epsilon$-good.

  Since $\iQfct{\cdot\,}$ is decreasing, for a given $\epsilon$ the
  family $\delta$ is $\epsilon$-good if and only if
  \begin{equation*}
    \deltan \geq \sigma \sqrt{n} \, \iQfct{ \frac{\epsilonn}{2} \, } \qquad \text{for } n \text{ large enough}.
  \end{equation*}

\end{example*}

Through its definition, one can readily derive the following
properties of $\delta$-typical sets.
\begin{property}
  \label{prop:KLTset_1}
  The set $\KLTset$ satisfies the following properties:
  \begin{itemize}
  \item[(a)] For $\vx \in \KLTset$,
    \begin{align*}
      \pn(\vx) \, 2^{- \left( \Dn + \deltan \right)} & \leq \qn(\vx) \leq \pn(\vx) \, 2^{- \left( \Dn - \deltan \right)}.
    \end{align*}
    
  \item[(b)] $\qn \! \left( \KLTset \right) \leq 2^{- \left( \Dn - \deltan \right)}$.
  \end{itemize}
\end{property}

For $\epsilon$-good families~\eqref{eq:KLdeltaprop}, the
$\delta$-typical set satisfies the following additional properties.
\begin{property}
  \label{prop:KLTset}
  For an $\epsilon$-good family $\delta$, the set $\KLTset$ satisfies
  the following properties:
  \begin{itemize}
  \item[(a)] For $n$ is large enough,
    \begin{equation*}
      \pn \! \left( \KLTset \right) \geq \left( 1 - \epsilonn \right).
    \end{equation*}
    
  \item[(b)] For $n$ is large enough,
    \begin{equation}
      \qn \! \left( \KLTset \right) \geq \left(1 - \epsilonn \right) \, 2^{- \left( \Dn +
	  \deltan \right)}.
      \label{eq:qnTSet}
    \end{equation}
  \end{itemize}
\end{property}

These properties state in layman terms that whenever $\delta$ is
$\epsilon$-good, the probability --under $\pn$-- of the
$\delta$-typical set is $\epsilonn$ close to one and under $\qn$ it is
lower bounded by~\eqref{eq:qnTSet}. In the following lemma we consider
another "set" that has large probability under $\pn$ and we show
that its probability under $\qn$ is lower-bounded by a properly
adjusted lower bound.

\begin{lemma}
  \label{lm:KLOtherSets}
  Let $B^{[n]} \subset \set{X}^n$ be any subset of sequences
  $\vx \in \set{X}^n$ such that
  $\pn \! \left( B^{[n]} \right) \geq \left( 1 - \epsn \right)$ when
  $n$ is large enough, and where $\epsn$ is a family of positive scalars less than one. For an $\epsilon$-good $\delta$ and $\qn$ such
  that $\Dn < \infty$,
  \[ \qn \! \left( B^{[n]} \right) \geq \left( 1 - \epsilonn - \epsn \right) \, 2^{-\left( \Dn +
      \deltan \right)},\]
  whenever $n$ is large enough.
\end{lemma}

\begin{proof}
  Consider $n$ to be large enough for the inequalities assumed in the
  statement and in Property~\ref{prop:KLTset} to hold.

  Since
  $\pn \! \left( B^{[n]} \right) \geq \left( 1 - \epsn \right)$
  and
  $\pn \! \left( A^{[n]} \right) \geq \left( 1 - \epsilonn \right)$,
  by the union bound
  \begin{equation*}
    \pn \! \left( A^{[n]} \cap B^{[n]} \right) \geq \left( 1 - \epsilonn - \epsn \right).
  \end{equation*}
  
  Therefore,
  \begin{align*}
    \qn \! \left( B^{[n]} \right)
    & \geq \, \qn \! \left( A^{[n]} \cap B^{[n]} \right) 
    = \!\! \int_{\vx \in A^{[n]} \cap B^{[n]} } \qn (\vx) \, d \vx \\
    & \geq \int_{\vx \in A^{[n]} \cap B^{[n]} } \pn(\vx) \, 2^{- \left( \Dn + \deltan \right)} \\
    & = \, 2^{- \left( \Dn + \deltan \right)} \pn \! \left( A^{[n]} \cap B^{[n]} \right) \\
    & \geq \, \left( 1 - \epsilonn - \epsn \right) \, 2^{- \left( \Dn + \deltan \right)}.
  \end{align*}

  Naturally. the derivations hold for the discrete case with integrals
  replaced by summations.
\end{proof}

\section{Generalized Chernoff–Stein Lemma}
\label{sec:Main}

In this section we state and prove an equivalent result to Stein's
Lemma for the continuous case. The stated result is also an extension
of Stein's Lemma to setups where variables are also possibly
correlated. We start by stating and proving a relevant lemma.

\begin{lemma}
  \label{lma:BNDS}
  Consider the binary hypothesis test: $X_1, X_2, \cdots, X_n$
  distributed according to $\pn$ or $\qn$, where $\Dn$ is {\em
    finite\/}. Let $\set{B}_{p} \subset \set{X}^n$ be the decision
  region for hypothesis $\pn$. Denote the probabilities of error by
  \begin{equation*}
    \alpha^{[n]} = \pn \! \left( \set{B}_p^{c} \right) \qquad \beta^{[n]} = \qn \! \left( \set{B}_p \right),
  \end{equation*}
  and for $0 < \tau < \frac{1}{2}$, we abuse notation and denote by
  $\tau$ the family of constants: $\taun = \tauone = \tau$ and
  define
  \begin{equation*}
    \beta^{[n]}_{\tau} = \min_{\stackrel{\set{B}_{p} \subset \set{X}^n}{\alpha^{[n]} < \tau}} \beta^{[n]}.
  \end{equation*}
  
  For any $\tau$-good family $\gamma$ and $\epsilon$-good family
  $\delta$,
  \begin{equation*}
    \left( 1 - \epsilonn - \tau \right) 2^{-\left( \Dn + \deltan \right)} \leq \beta^{[n]}_{\tau} \leq 2^{- \left( \Dn - \gamman \right)}.
  \end{equation*}
\end{lemma}

\begin{proof}
  The proof follows the same path as in the discrete and independent
  case studied in Cover~\cite{cover}. 

  By considering $\tau$-good $\gamma$ families, we derive first the
  upperbound by considering the detector
  $\set{B}_p = \KLTsetd{\gamma}$. Since $\gamma$ is $\tau$-good,
  Property~\ref{prop:KLTset} implies that for $n$ large enough,
  \begin{equation*}
    \alpha^{[n]} = \pn \! \left( \KLTsetd{\gamma}^{c} \right) \leq \taun = \tau,
  \end{equation*}
  which means that the device is feasible. Additionally,
  Property~\ref{prop:KLTset_1} implies,
  \begin{align*}
    \beta^{[n]} & = \qn \! \left( \KLTsetd{\gamma} \right) \leq 2^{- \left( \Dn - \gamman \right)} \\
    \Leftrightarrow \quad - \log \beta^{[n]} & \geq \left( \Dn - \gamman \right).
  \end{align*}
   
  An optimal device will outperform the above device and hence:
  \begin{equation}
    - \log \beta^{[n]}_{\tau}  \geq \left( \Dn - \gamman \right).
    \label{eq:lowerBnd}
  \end{equation}
  
  On the other hand, {\em for any device\/} such that
  $\pn \! \left( \set{B}_p^{c} \right) < \tau$,
  $\pn \! \left( \set{B}_p \right) > (1 - \tau)$ and given a $\delta$ that is $\epsilon$-good, Lemma~\ref{lm:KLOtherSets} applies:
  \begin{equation*}
    \beta^{[n]} = \qn \! \left( \set{B}_p \right) \geq \left( 1 - \epsilonn - \tau \right) \, 2^{- \left( \Dn + \deltan \right)}.
  \end{equation*}
  
  Hence
  \begin{align*}
    \log \beta^{[n]} & \geq \log \left( 1 - \epsilonn - \tau \right) - \left( \Dn + \deltan \right) \\
    - \log \beta^{[n]} & \leq \Dn - \log \left( 1 - \epsilonn - \tau \right) + \deltan.
  \end{align*}

  This applies to all devices, including an optimal one. Hence,
  \begin{equation}
    - \log \beta^{[n]}_{\tau}  \leq \Dn - \log \left( 1 - \epsilonn - \tau \right) + \deltan.
    \label{eq:upperBnd}
  \end{equation}
   
\end{proof}

Whenever we can find good $\gamma$ and $\delta$ families with
additional structure, one can readily prove additional properites that
we state in the following (comprehensive) theorem.


\begin{theorem}[Generalized Chernoff–Stein Lemma]
  \label{thm:GCSL}
  Consider the binary hypothesis test: $X_1, X_2, \cdots, X_n$
  distributed according to $\pn$ or $\qn$, where $\Dn$ is {\em
    finite\/}. Let $\set{B}_{p} \subset \set{X}^n$ be the decision
  region for hypothesis $\pn$ and denote by
  \begin{equation*}
    \begin{array}{ll}
      \displaystyle \alpha^{[n]} = \pn \! \left( \set{B}_p^{c} \right) \qquad & \beta^{[n]} = \qn \! \left( \set{B}_p \right) \vspace{7pt} \\
      \displaystyle \beta^{[n]}_{\tau} = \min_{\stackrel{\set{B}_{p} \subset \set{X}^n}{\alpha^{[n]} < \tau}} \beta^{[n]} & 0 < \tau < \frac{1}{2},
    \end{array}
  \end{equation*}

  If
  \begin{itemize}
  \item $\Dn$ {\em grows with $n$\/} and
  \item For any $\xi > 0$ there exists a $\tau$-good $\delta$ such that
    $\deltan < \xi \Dn$ for $n$ large enough,
  \end{itemize}
  then
  \begin{equation*}
    - \log \beta^{[n]}_{\tau} = \Dn + o \left( \Dn \right).
  \end{equation*}
\end{theorem}


Before proceeding, we note that if there exists a good $\delta$ such
that $\deltan = o\left(\Dn\right)$, then the condition holds and
\begin{equation*}
  - \log \beta^{[n]}_{\tau} = \Dn + o \left( \Dn \right).
\end{equation*}

\begin{proof}
  The theorem is readily established by using Lemma~\ref{lma:BNDS},
  \begin{equation*}
    \Dn - \deltan  < - \log \beta^{[n]}_{\tau} < \Dn + \deltan - \log (1 - 2 \tau ).
  \end{equation*}

  Since $\Dn$ grows with $n$ and for every $\xi > 0$ one can chose
  $\delta$ such $\deltan \leq \xi \Dn$:
  \begin{equation*}
    - \log \beta^{[n]}  = \Dn + o \left( \Dn \right).
  \end{equation*}
\end{proof}

\begin{example*}[IID]
  In that scenario,
  \begin{equation*}
    \Dn = n \Done
  \end{equation*}
  and grows linearly with $n$. Since
  \begin{align*}
    &\frac{1}{n} \sum_{k=1}^n \log \frac{\pone(X_k)}{\qone(X_k)} \xrightarrow[n\to\infty]{p} \Done,
  \end{align*}
  then for any $\xi > 0$, the family
  $\left\{ \deltan = n \, \xi \right\}$ is good. The requirements
  are hence satisfied and
  \begin{equation*}
    - \frac{1}{n} \log \beta^{[n]} \xrightarrow[n\to\infty]{} \Done,
  \end{equation*}
  recovering the textbook~\cite{cover} result in the discrete IID
  setup.
\end{example*}


\section{Testing the Correlated Gaussian Hypotheses}
\label{sec:GG}

We apply the results derived in~\ref{sec:Main} to the problem of binary hypothesis testing between two possibly correlated Gaussian laws. More specifically, we consider two discrete-time zero-mean (wide-sense) stationary Gaussian processes with absolutely summable auto-covariance functions and spectra denoted as:
\begin{align*}
  H_0: \,\, & K[n] = K_p[n] \quad \leftrightarrow \quad S_p \dtf \\
  H_1: \,\, & K[n] = K_q[n] \quad \leftrightarrow \quad S_q \dtf.
\end{align*}
Note that since the functions $K_p[\cdot]$ and $K_q[\cdot]$ are absolutely summable, the spectra are guaranteed to exist and will be clearly upper bounded. We impose the additional restriction that these spectra are lower-bounded by some positive scalars. This condition is for example satisfied whenever the spectra are positive and continuous on $[0,1]$.

We observe $n$ consecutive values of the process:
\begin{equation*}
    \vY = \begin{pmatrix}
    Y[1] \\ \vdots \\ Y[n]
    \end{pmatrix},
\end{equation*}
and our objective is to test for the two hypotheses.

Before proceeding, we note that under either hypethesis the auto-covariance matrix of $\vY$ is the Hermitian Toeplitz matrix derived form the auto-covariance function $K[\cdot]$:
\begin{equation*}
    \Lambda = \begin{pmatrix}
	K[0] & K[1] & \cdots & \cdots & K[n-2] & K[n-1] \\
	K[1] & K[0] & K[1] & \cdots & \cdots & K[n-2] \\
	\cdots & \ddots & \ddots & \ddots & \cdots & \cdots \\
	\cdots & \cdots & \ddots & \ddots & \ddots & \cdots \\
	 \cdots&\cdots&\cdots & K[1] & K[0] & K[1] \\
	K[n-1] & \cdots&\cdots&\cdots & K[1] & K[0]
    \end{pmatrix}
\end{equation*}

\subsection{Equivalent Formulation}
\label{sec:PS}

When observing $n$ consecutive values of the process, the binary hypothesis test at hand becomes one where the observed data is simply an $n$-dimensional multivariate Gaussian vector:
\begin{align}
  H_0: \, & \vY \sim \Normal{\vzero}{\Lambda_p} \nonumber \\
  H_1: \, & \vY \sim \Normal{\vzero}{\Lambda_q}. \label{eq:GGsetup}
\end{align}

By applying a --well chosen-- linear transformation of the form
$\mat{M} \vY$, the detection problem is equivalent to one with
diagonal covariance matrices.

Indeed, we first diagonalize matrix $\Lambda_q$ using orthogonal
matrix $\mat{U}_q$. Denoting the diagonal matrix by $\Delta_q$ and using the convention as in equation~\eqref{eq:matSQRT},
\begin{equation*}
  \begin{array}{rlcrl}
   \Lambda_q \bs & = \mat{U}_q^t \Delta_q \mat{U}_q & \qquad & \Lambda_q^{1/2} \bs & = \mat{U}_q^t \Delta^{1/2}_q \mat{U}_q,
  \end{array}
\end{equation*}
consider the matrix $\Lambda_q^{-1/2} \Lambda_p \Lambda_q^{-t/2}$
which is Hermitian symmetric and positive definite. We diagonalize it using an orthogonal matrix $\mat{V}$ and denote the resulting diagonal matrix
$\mat{K}$:
\begin{align*}
  \Lambda_q^{-1/2} \Lambda_p \Lambda_q^{-t/2}  = \mat{V}^t \mat{K} \, \mat{V}.
\end{align*}

Finally, consider the transformation:
\begin{equation*}
  \vY' = \mat{V} \Lambda_q^{-1/2} \vY,
\end{equation*}
which under $H_0$ and $H_1$ has respective Gaussian laws $\pnp$ and
$\qnp$ with respective covariance matrices:
\begin{align}
  \Lambda_{p'} & =  \mat{V} \Lambda_q^{-1/2} \Lambda_p \Lambda_q^{-t/2} \, \mat{V}^t = \mat{K} \label{eq:matrixK} \\
  \Lambda_{q'} & =  \mat{V} \Lambda_q^{-1/2} \Lambda_q \Lambda_q^{-t/2} \, \mat{V}^t = \mat{I}. \nonumber
\end{align}

To recap, by applying the linear transformation
$\mat{V} \Lambda_q^{-1/2} \vY$, the detection problem becomes
\begin{align}
  H_0: \, & \vY \sim \Normal{\vzero}{\mat{K}} \nonumber\\
  H_1: \, & \vY \sim \Normal{\vzero}{\mat{I}}. \label{eq:GGsetuequiv}
\end{align}

Note here that $\KL{\pn}{\qn} = \KL{\pnp}{\qnp}$
and there is no loss in generality to assume in what follows that
\begin{equation*}
  \pn = \Normalf{\vx}{\vzero}{\mat{K}} \quad \& \quad \qn = \Normalf{\vx}{\vzero}{\mat{I}}.
\end{equation*}

\begin{note}
  By inspection of equation~\eqref{eq:matrixK} the elements of the diagonal matrix $\mat{K}$ are the eigenvalues of the matrix $\Lambda_q^{-1/2} \Lambda_p \Lambda_q^{-t/2}$.
  By considering an eigenvector associated with an eigenvalue and applying an appropriate matrix multiplication, one can readily show
  that the eigenvalues of the following matrices are identical:
  \begin{equation*}
    \Lambda_q^{-1/2} \Lambda_p \Lambda_q^{-t/2} \qquad \Lambda_q^{-1} \Lambda_p \qquad \Lambda_p \Lambda_q^{-1}
    \qquad \Lambda_p^{1/2} \Lambda_q^{-1} \Lambda_p^{t/2}.
  \end{equation*}  
  \label{note:Eig}
\end{note}

\subsection{Main Result}

We state and prove our main result in the form of Theorem~\ref{Gaussian-Gaussian}. 

\begin{theorem}
\label{Gaussian-Gaussian}
  Consider the binary hypothesis test where we observe an $n$-dimensional "segment" $\{X_1, X_2, \cdots, X_n\}$ of a zero-mean wide-sense stationary Gaussian process with an auto-covariance function $K_p[\cdot]$ or $K_q[\cdot]$. We assume that  $K_p[\cdot]$ and $K_q[\cdot]$ are absolutely summable and their Fourier transforms (i.e. spectra) $S_p \dtf$ and $S_q \dtf$ are upper-bounded and lower-bounded by a positive scalar. The observation $\vX$ is hence distributed according to $\pn \sim \Normal{\vzero}{\Lambda_p}$ or $\qn \sim \Normal{\vzero}{\Lambda_q}$. Let $\set{B}_{p} \subset \set{X}^n$ be the decision region for hypothesis $\pn$ and denote by
  \begin{equation*}
    \begin{array}{ll}
      \displaystyle \alpha^{[n]} = \pn \! \left( \set{B}_p^{c} \right) \qquad & \beta^{[n]} = \qn \! \left( \set{B}_p \right) \vspace{7pt} \\
      \displaystyle \beta^{[n]}_{\tau} = \min_{\stackrel{\set{B}_{p} \subset \set{X}^n}{\alpha^{[n]} < \tau}} \beta^{[n]} & 0 < \tau < \frac{1}{2},
    \end{array}
  \end{equation*}

  Then, 
  \begin{itemize}
  \item[1)] The relative entropy  $\Dn = \KL{\pn}{\qn}$ grows linearly  with $n$: 
    \begin{equation*}
    \frac{1}{n} \Dn \xrightarrow[n \to \infty]{} C_s,
    \end{equation*}
  where 
  \begin{equation*}
    C_s \eqdef \frac{1}{2} \!\int_0^1 \left(\frac{S_p \dtf}{S_q \dtf }  - \log \frac{S_p \dtf}{S_q \dtf }  - 1\right) df.
  \end{equation*}

  \item[2)] The type II error decays linearly with $n$ to the first order of the exponent: 
    \begin{equation*}
    - \log \beta^{[n]}_{\tau} = C_s \, n + o \left( n \right).
    \end{equation*}
  \end{itemize}
\end{theorem}

The proof of the first part of the theorem is deferred to Appendix~\ref{appAsymRE}. The second part follows directly from the results of Theorem~\ref{thm:GCSL} whenever one can establish the existence of a $\tau$-good family $\delta$ such that $\deltan < \xi \Dn$ for $n$ large enough, for any $\xi >0$. The rest of this section is dedicated to establishing the existence of such a family.

By the analysis in Section~\ref{sec:PS} an equivalent problem is defined by equation~\eqref{eq:GGsetuequiv} where $\mat{K} = \text{diag}\left( \kappa^{[n]}_1, \cdots ,\kappa^{[n]}_n \right)$ and as shown in Appendix~\ref{appD}, the relative entropy is 
\begin{align*}
 \Dn & = \KL{ \pn}{\qn} = \frac{1}{2} \text{tr}(\mat{K}) - \frac{1}{2} \log \text{det}(\mat{K}) - \frac{n}{2} \\
 & =  \frac{1}{2} \sum_{k=1}^n \kappa^{[n]}_1 - \frac{1}{2} \sum_{k=1}^n \log \kappa^{[n]}_1  - \frac{n}{2},
\end{align*}
which implies
\begin{align}
    & \pn \left( \KLTset \right) \nonumber\\
    & = \pn \left( - \deltan \leq \log
      \frac{\pn(\vX)}{\qn(\vX)} - \Dn \leq  \deltan\right) \nonumber\\ 
    & = \pn \left( - \deltan \leq \frac{1}{2}\sum_{k=1}^n \left[ 1 - \frac{1}{\kappa^{[n]}_k} \right] X_k^2  - \frac{1}{2} \sum_{k=1}^n \log \kappa^{[n]}_k \right. \nonumber\\
    & \hspace{180pt} \left.  - \Dn \leq \deltan \right) \nonumber \\
    & = \pn \left( - 2 \deltan \leq \sum_{k=1}^n \left[ 1 - \frac{1}{\kappa^{[n]}_k} \right] \left[ X_k^2 - \kappa^{[n]}_k \right] \leq 2\deltan \right) \nonumber \\
    & = \pn \left( - 2 \deltan \leq \sum_{k=1}^n \left[ \kappa^{[n]}_k - 1 \right] \! \left[ \frac{X_k^2}{\kappa^{[n]}_k} - 1 \right] \leq 2\deltan \right) 
    \label{eq:nonidsum}
  \end{align}

Before proceeding with the analysis of equation~\eqref{eq:nonidsum}, we define 
\begin{equation*}
    B_{n_{k}} = \left[ \kappa^{[n]}_k - 1 \right],\qquad  B_n^2 =\sum_{k = 1}^n B^2_{n_{k}},
\end{equation*}
and
\begin{align}
& \pn \left( - 2 \deltan \leq \sum_{k=1}^n \left[ \kappa^{[n]}_k - 1 \right] \! \left[ \frac{X_k^2}{\kappa^{[n]}_k} - 1 \right] \leq 2\deltan \right) \nonumber\\
& = \pn \left( \! - \frac{\sqrt{2} \deltan}{ B_n} \leq \frac{1}{\sqrt{2} B_n} \sum_{k=1}^n B_{n_k} \!\! \left[ \frac{X_k^2}{\kappa^{[n]}_k} - 1 \right] \leq \frac{\sqrt{2} \deltan}{B_n} \right) \nonumber\\
& = \pn \left( -  \frac{\sqrt{2} \deltan}{ B_n} \leq \sum_{k=1}^n \psi_{n_{k}} \leq \frac{\sqrt{2} \deltan}{B_n} \right), \label{eq:CLTgen}
\end{align}
where for $k = 1, \cdots, n$,
\begin{equation}
    \psi_{n_{k}} \eqdef \frac{B_{n_{k}}}{\sqrt{2}B_n} \left[ Y_k^2 - 1 \right], \quad Y_k = \frac{X_k}{\sqrt{\kappa^{[n]}_k}}.
    \label{eq:xiDef}
\end{equation}
Note that since for any $k$, $Y_k$ is a zero-mean Gaussian variable with variance equal to one, its square is central chi-squared distributed with one degree of freedom with mean and variance equal to 1 and 2 respectively.

We recall the following result on the general limit of sums~\cite{kolmo}.

\begin{theorem}
\label{th:kolmo}
In order that for suitably chosen constants $A_n$ the distribution functions of the sums 
\begin{equation*}
    \zeta_n = \psi_{n_{1}} + \psi_{n_{2}} + \cdots + \psi_{n_{k_n}} - A_n
\end{equation*}
of independent zero-mean random variables $\psi_{n_{1}}, \psi_{n_{2}}, \cdots, \psi_{n_{k_n}}$ such that 
$\sum_{k=1}^{k_n} \text{E}[\psi^2_{n_{k}}] = 1$
converge to the normal law 
\begin{equation}
\Phi(x) = \frac{1}{\sqrt{2 \pi}} \int_{-\infty}^{x} e^{- \frac{z^2}{2}} \, dz,\label{eq:CDFn}
\end{equation}
and the variables $\psi_{n_{1}}, \psi_{n_{2}}, \cdots, \psi_{n_{k_n}}$ be infinitesimal, it is necessary and sufficient that for every $\epsilon > 0$
\begin{equation}
\sum_{k=1}^{k_n} \int_{|x| > \epsilon} x^2 dF_{\psi_{n_k}} (x) \underset{n \rightarrow \infty}{\longrightarrow} 0 \label{eq:iffcond}
\end{equation}
  \end{theorem}
\begin{proof}
See~\cite[Thm. 3, p. 101]{kolmo}
\end{proof}

In our setting, the $\{ \psi_{n_{k}} \}$'s defined in~\eqref{eq:xiDef} are independent shifted and scaled degree-one chi-squared variables with $\E{\psi_{n_{k}}} = 0$ and $\sum_{k=1}^{n} \E{\psi^2_{n_{k}}} = 1$.


Additionally, $B_{n_{k}} = \left[ \kappa^{[n]}_k - 1 \right]$ and by Theorem~\ref{th:contfct} and equation~(\ref{eq:CovMatProduct}) (which is a result of the fact that the spectra are upper-bounded and positively lower-bounded), they are uniformly bounded for all $k$ and $n$.
When it comes to $B_n = \sqrt{ \sum_{k = 1}^n \left[ \kappa^{[n]}_k - 1 \right]^2} $,
\begin{align}
  \frac{1}{\sqrt{n}} B_n & = \sqrt{ \frac{1}{n} \sum_{k = 1}^n \left[ \kappa^{[n]}_k - 1 \right]^2} \\
&   \xrightarrow[n \to \infty]{} \sqrt{ \int_0^1  \left[ \frac{S_p \dtf}{S_q \dtf } -1 \right]^2 df } \label{eq:bnsim},
\end{align}
due to uniform convergence and following similar arguments to what is done in Appendix~\ref{appDC}. Therefore, for every $\epsilon > 0$, one can choose $n$ large enough such that $\left|\frac{B_{n_{k}}}{\sqrt{2} B_n}\right| < \epsilon$, and each term in equation~\eqref{eq:iffcond} simplifies to: 
\begin{align}
 & \int_{|x| > \epsilon} x^2 dF_{\psi_{n_{k}}} (x) \nonumber\\
 &\, = \hspace{-1pt}\int_{\epsilon}^\infty \!\! \hspace{-0.6pt} x^2 \frac{1}{\sqrt{2 \pi}} \left|\frac{\sqrt{2} B_n}{B_{n_{k}}}\right| \left( \left|\frac{\sqrt{2}B_n}{B_{n_{k}}}\right|x + 1\right)^{-\frac{1}{2}} \!\! e^{-\frac{\left|\frac{\sqrt{2}B_n}{B_{n_{k}}}\right|x + 1}{2}}\, dx \label{eq:pdfchi}\\
 &\, \leq \frac{e^{-\frac{1}{2}}}{\sqrt{2 \pi}} \left|\frac{\sqrt{2}B_n}{B_{n_{k}}}\right|^{\frac{1}{2}} \int_{\epsilon}^\infty x^{\frac{3}{2}}  e^{-\left|\frac{B_n}{\sqrt{2} B_{n_{k}}}\right|x}\, dx \nonumber\\
 &\, = \frac{2 e^{-\frac{1}{2}}}{\sqrt{\pi}} \left(\frac{B_{n_{k}}}{B_n}\right)^2 \int_{\frac{\epsilon}{\sqrt{2}} \left|\frac{ B_n}{B_{n_{k}}}\right|}^\infty u^{\frac{3}{2}} e^{-u} \, du \label{eq:covux}\\
 & \, = \frac{2 e^{-\frac{1}{2}}}{\sqrt{\pi}} \left(\frac{B_{n_{k}}}{B_n}\right)^2 \Gamma\left(\frac{5}{2},\frac{\epsilon}{\sqrt{2}}\left|\frac{B_n}{B_{n_{k}}}\right|\right), \label{eq:uppincomgam}
 \end{align}
 where in order to write equation~\eqref{eq:pdfchi}, we used the expression of the PDF of $Y_k^2$: 
 \begin{equation*}
 p_{Y_{k}^2}(x) = \frac{1}{\sqrt{2 \pi}} x^{-\frac{1}{2}} e^{-\frac{x}{2}}.
 \end{equation*}
 Equation~\eqref{eq:covux} is justified by the change of variable $u = \left|\frac{B_n}{ \sqrt{2} B_{n_{k}}}\right| x$ and $\Gamma(\cdot,\cdot)$ in equation~\eqref{eq:uppincomgam} denotes the upper incomplete Gamma function. Therefore, the LHS of equation~\eqref{eq:iffcond} is
 \begin{align}
 &\sum_{k=1}^{n} \int_{|x| > \epsilon} x^2 dF_{\psi_{{n}_k}} (x) \nonumber \\
 &\, \leq \frac{2 e^{-\frac{1}{2}}} {\sqrt{\pi}} \sum_{k=1}^{n} \left(\frac{B_{n_{k}}}{B_n}\right)^2 \Gamma\left(\frac{5}{2},\frac{\epsilon}{\sqrt{2}}\left|\frac{B_n}{B_{n_{k}}}\right|\right) \nonumber\\
 &\, \leq \frac{4 e^{-\frac{1}{2}}} {\sqrt{\pi}} \left(\frac{\epsilon}{\sqrt{2}}\right)^{\frac{3}{2}} \sum_{k=1}^{n} \left|\frac{B_{n_{k}}}{B_n}\right|^{\frac{1}{2}} e^{-\frac{\epsilon}{\sqrt{2}}\left|\frac{B_n}{B_{n_{k}}}\right|}\label{eq:gamsim} \\
 &\, \leq \frac{4 e^{-\frac{1}{2}}} {\sqrt{\pi}} \left(\frac{\epsilon}{\sqrt{2}}\right)^{\frac{3}{2}} \left(\frac{M_s}{B_n}\right)^{\frac{1}{2}} e^{-\frac{\epsilon}{\sqrt{2}}\frac{B_n}{M_s}} \, n\label{eq:supspe} \\
 &\, \overset{n \rightarrow \infty}{\longrightarrow}  0, \label{eq:lim0}
 \end{align}
where in order to write equation~\eqref{eq:gamsim} we used the fact that $\left|\frac{B_n}{B_{n_{k}}}\right|$ is large enough and that
\begin{eqnarray*}
\Gamma\left(\frac{5}{2},x\right) &\sim& x^{\frac{3}{2}} e^{-x}, \quad x \rightarrow \infty\\
&\leq& 2 \, x^{\frac{3}{2}} e^{-x}
\end{eqnarray*}
Equation~\eqref{eq:supspe} is justified by the fact that $\{B_{n_{k}}\}$s, $1 \leq k \leq n$ are uniformly bounded by $M_s$ 
and equation~\eqref{eq:lim0} is valid since $B_n = \Theta\left(\sqrt{n}\right)$ as given by equation~\eqref{eq:bnsim}.  

In conclusion, equation~\eqref{eq:iffcond} is justified and Theorem~\ref{th:kolmo} applies with $A_n = 0$. Hence, 

 \begin{equation*}
    \sum_{k=1}^n \psi_{n_{k}}\xrightarrow[n\to\infty]{d}
    \Normal{0}{1},
  \end{equation*}
  and equation~\eqref{eq:CLTgen} gives
\begin{align}
 \pn \left( \KLTset \right) \rightarrow \, \, 1 - 2 \, \Qfct{ \frac{\sqrt{2} \deltan}{ B_n} }.\label{eq:thkolmo}
\end{align}  

Equation~\eqref{eq:thkolmo} implies $\delta$ is $\epsilon$-good if and only if
\begin{equation*}
    \deltan \geq  \frac{B_n}{\sqrt{2}} \, \iQfct{ \frac{\epsilonn}{2} \, } \qquad \text{for } n \text{ large enough}.
  \end{equation*}

Since $B_n = \Theta(\sqrt{n})$ then finding an $\epsilon$-good $\delta$ where $\deltan < \xi \Dn$ for any $\xi >0$ is feasible because $\Dn = \Theta(n)$ as given by the first part of the theorem. 
\bibliographystyle{IEEEtran}
\bibliography{Reference}
\appendices

\section{Asymptotics of Functions of Eigenvalues of Covariance Matrices}
\label{app:FctEigCov}

Using properly constructed asymptotically equivalent matrices, we derive below a primary result in the form of Theorem~\ref{th:CovMatEV} where we determine the asymptotic behavior of continuous functions of the eigenvalues of an auto-covariance matrix.

\subsection{Asymptotically Equivalent Sequences of Matrices}
\label{appEqSeq}

We list the definition and some relevant results in the theory of asymptotically equivalent sequences of matrices. We refrain from providing the proofs that can be found in the review by Gray~\cite{Gray} for example.

\begin{definition}

Let $\left\{\An\right\}$ and $\left\{\Bn\right\}$ be two sequences of $n \times n$ matrices. The two sequences are said to be asymptotically equivalent if
\begin{enumerate}
\item $\An$ and $\Bn$ are uniformly bounded in the "strong" $l_{2}$ norm:
\begin{equation*}
   \twonorm{\An}, \twonorm{\Bn} \leq K < \infty.
\end{equation*}

\item $\An - \Bn = \Dn$ goes to zero in the "weak" norm\footnote{We will denote the weak norm of a matrix $\mat{A}$ by the symbol $\norm{\mat{A}}$ defined in the following manner: $\norm{\mat{A}} = \sqrt{
\frac{1}{n} \sum_k \sum_j |a_{k,j}|^2}$ } as $n \to \infty$:
\begin{equation*}
   \lim_{n \to \infty} \norm{ \An - \Bn } = \lim_{n \to \infty} \norm{\Dn} = 0.
\end{equation*}

Asymptotic equivalence of $\left\{\An\right\}$ and $\left\{\Bn\right\}$ will be abbreviated
$\An \sim \Bn$.
\end{enumerate}
\end{definition}

\begin{property} \mbox{} \\
\begin{enumerate}
\item If $\An \sim \Bn$ and if $\Bn \sim \Cn$, then $\An \sim \Cn$.
\item If $\An \sim \Bn$ and if $\Cn \sim \Dn$, then $\An \Cn \sim \Bn \Dn$.
\item If $\An \sim \Bn$ and $\norm{\Ani}, \norm{\Bni} \leq K < \infty$, i.e., $\Ani$ and $\Bni$ exist and are uniformly bounded by some constant independently of $n$, then $\Ani \sim \Bni$.
\end{enumerate}
\label{th:asymeq}
\end{property}


\begin{theorem}

Let $\left\{\An\right\}$ and $\left\{\Bn\right\}$ be asymptotically equivalent sequences of matrices with real eigenvalues $\left\{ \alpha^{[n]}_{k} \right\}$ and $\left\{ \beta^{[n]}_{k} \right\}$ respectively. There exist finite numbers $m$ and $M$ such that
\begin{equation*}
  m \leq \alpha^{[n]}_{k}, \beta^{[n]}_{k} \leq M, \quad n = 1, 2, \cdots \quad k = 0,1,\cdots, n-1.
\end{equation*}

Let $F(x)$ be an arbitrary function continuous on $[m,M]$.  Then
\begin{equation*}
  \lim_{n \to \infty} \frac{1}{n} \sum_{k=0}^{n-1} F\left(\alpha^{[n]}_{k} \right) = \lim_{n \to \infty} \frac{1}{n} \sum_{k=0}^{n-1} F \left(\beta^{[n]}_{k} \right).
\end{equation*}
if either of the limits exists.

\label{th:contfct}
\end{theorem}

We note that the statement of the Theorem in~\cite{Gray} is for Hermitian matrices. However, a close examination of the derivation of the results readily indicates that the "Hermitian" requirement is simply made to guarantee that the eigenvalues are real-valued allowing the application of the Weierstrass approximation theorem for the uniform approximation of continuous functions on closed intervals with polynomials.

\subsection{Some Asymptotic Equivalence of Covariance Matrices}
\label{appCovM}

Let $\{X[\cdot]\}$ be a discrete-time wide-sense stationary process that is zero mean and with auto-covariance function $K[m]$. We assume throughout that $K[\cdot]$ is absolutely summable.

When considering an $n$-dimensional vector 
\begin{equation*}
    \vect{X} = \begin{pmatrix}
    X[1] \\ \vdots \\ X[n]
    \end{pmatrix},
\end{equation*}
its auto-covariance matrix is the Hermitian Toeplitz matrix:
\begin{equation*}
    \Ln = \begin{pmatrix}
	K[0] & K[1] & \cdots & \cdots & K[n-2] & K[n-1] \\
	K[1] & K[0] & K[1] & \cdots & \cdots & K[n-2] \\
	\cdots & \ddots & \ddots & \ddots & \cdots & \cdots \\
	\cdots & \cdots & \ddots & \ddots & \ddots & \cdots \\
	 \cdots&\cdots&\cdots & K[1] & K[0] & K[1] \\
	K[n-1] & \cdots&\cdots&\cdots & K[1] & K[0]
    \end{pmatrix}.
\end{equation*}

Consider the matrix
\begin{equation*}
    \Lnh = \begin{pmatrix}
	K[0] & \cdots & K[\hat{n}-1] & 0 & 0 \\
	K[1] & K[0] & \cdots & K[\hat{n}-1] & 0 \\
	\cdots & \ddots & \ddots & \cdots & \cdots \\
	K[\hat{n}-1] & \cdots & \ddots & \ddots & \cdots \\
	 0 &K[\hat{n}-1]&\cdots & K[0] & K[1] \\
	0 & 0 &\cdots  & K[1] & K[0]
    \end{pmatrix},
\end{equation*}
where $\hat{n} = \lfloor n/2 \rfloor + 1$. We argue that $\Ln \sim \Lnh$. Indeed,
\begin{itemize}
\item The matrices are uniformly bounded in the $l_{2}$ norm. Indeed, for any matrix $\mat{M}$
\begin{equation*}
    \twonorm{\mat{M}} \leq \sqrt{ \norm{\mat{M}}_{1} \norm{\mat{M}}_{\infty} }.
\end{equation*}
Since the $l_{1}$ and $l_{\infty}$ norms are equal to the maximal absolute column and row sums respectively, whenever $\mat{M}$ is Hermitian, these $l_{1}$ and $l_{\infty}$ norms are equal and $\twonorm{\mat{M}} \leq \norm{\mat{M}}_{\infty}$. Therefore, 
\begin{equation*}
    \twonorm{ \Lnh } \leq \norm{ \Lnh }_{\infty} \quad \& \quad \norm{ \Ln }_{2} \leq \norm{ \Ln }_{\infty}.
\end{equation*}
Finally, the bound is readily obtained by using the absolute summability of the auto-covariance function:
\begin{align*}
    \norm{ \Lnh }_{\infty} & \leq \norm{ \Ln }_{\infty} = \max_{1 \leq i \leq n} \sum_{j=1}^n | \Lambda_{i,j} | \\
     & \leq 2 \sum_{j = 0}^{n-1} | K[j] | \leq 2 \sum_{j=0}^{\infty} | K[j] |.
\end{align*}

\item When it comes to the difference $\Ln - \Lnh$ is
\begin{align*}
    \begin{pmatrix}
	0 & 0 & \cdots & K[\hat{n}] & \cdots & K[n-1] \\
	0 & 0 & 0 & \cdots & \ddots & \cdots \\
	\cdots & \ddots & \ddots & \ddots & \cdots & K[\hat{n}] \\
	K[\hat{n}] & 0 & \ddots & \ddots & \ddots & \cdots \\
	 \cdots & \ddots & 0 & \cdots & \cdots & 0 \\
	K[n-1] & \cdots & K[\hat{n}] & 0 & \cdots  & 0 
    \end{pmatrix},
\end{align*}
and its weak norm is equal to 
\begin{align*}
  \norm{ \Ln - \Lnh }^2 & = 2 \, \frac{1}{n} \sum_{j = 1}^{n - \hat{n}} j |K[n-j]|^2 \\
  & \leq  2 \, \frac{1}{n} \sum_{j = 1}^{n - \hat{n}} (n - j) |K[n-j]|^2,
\end{align*}
because in the range $\{1, \cdots, n - \hat{n}\}$, $j \leq (n-j)$. Finally, since $\sum |K[j]|^2$ is convergent, and using Kronecker's lemma~\cite{ShiryBook}
\begin{align*}
    2 \, \frac{1}{n} \sum_{j = 1}^{n - \hat{n}} (n - j) |K[n-j]|^2 & =  2 \, \frac{1}{n} \sum_{j = \hat{n}}^{n - 1} j \, |K[j]|^2 \\
  & \leq 2 \, \frac{1}{n} \sum_{j = 1}^{n} j |K[j]|^2 \\
  & \xrightarrow[n \to \infty]{} 0
\end{align*}
\end{itemize}

Next we consider the Hermitian circulant matrices $\Lnt$ defined as
\begin{equation*}
    \begin{pmatrix}
	K[0] & \cdots & K[\hat{n}-1] & K[\hat{n}-2] & \cdots & K[1] \\
	K[1] & K[0] & \cdots & K[\hat{n}-1] & \cdots & K[2] \\
	\cdots & \ddots & \ddots & \cdots & \cdots & \cdots \\
	K[\hat{n}-1] & \cdots & \ddots & \ddots & \cdots & \cdots \\
    K[\hat{n}-2] & K[\hat{n}-1]&\cdots & \cdots & \cdots & \cdots \\
    \cdots & \ddots & \ddots & \cdots & \cdots & \cdots \\
    K[2] & \cdots &\cdots & \cdots & K[0] & K[1] \\
	K[1] & K[2] & \cdots &\cdots  & K[1] & K[0]
    \end{pmatrix}
\end{equation*}
when $n$ is even and
\begin{equation*}
    \begin{pmatrix}
	K[0] & \cdots & K[\hat{n}-1] & K[\hat{n}-1] & \cdots & K[1] \\
	K[1] & K[0] & \cdots & K[\hat{n}-1] & \cdots & K[2] \\
	\cdots & \ddots & \ddots & \cdots & \cdots & \cdots \\
	K[\hat{n}-1] & \cdots & \ddots & \ddots & \cdots & \cdots \\
    K[\hat{n}-1] & K[\hat{n}-1]&\cdots & \cdots & \cdots & \cdots \\
    K[\hat{n}-2] & \ddots & \ddots & \cdots & \cdots & \cdots \\
	\cdots & \ddots & \ddots & \cdots & \cdots & \cdots \\
	K[1] & \cdots & \cdots &\cdots  & K[1] & K[0]
    \end{pmatrix}
\end{equation*}
whenever $n$ is odd.

We note that $\Lnh \sim \Lnt$ as the conditions are satisfied:
\begin{itemize}
\item The matrices are uniformly bounded in the $l_{2}$ norm:
\begin{align*}
    \norm{ \Lnt }_{\infty} & \leq  2 \sum_{j =0}^{\hat{n} - 1} | K[j] | \leq 2 \sum_{j=0}^{\infty} | K[j] |,
\end{align*}
which is a finite constant.

\item Whenever $n$ is odd, the difference $\Lnh - \Lnh$ is
\begin{align*}
    \begin{pmatrix}
	0 & 0 & \cdots & K[\hat{n}-1] & \cdots & K[1] \\
	0 & 0 & 0 & \cdots & \ddots & \cdots \\
	\cdots & \ddots & \ddots & \ddots & \cdots & K[\hat{n}-1] \\
	K[\hat{n}-1] & 0 & \ddots & \ddots & \ddots & \cdots \\
	 \cdots & \ddots & 0 & \cdots & \cdots & 0 \\
	K[1] & \cdots & K[\hat{n}-1] & 0 & \cdots  & 0 
    \end{pmatrix}
\end{align*}
and its weak norm is equal to 
\begin{align*}
  \norm{ \Ln - \Lnh }^2 & = 2 \, \frac{1}{n} \sum_{j = 1}^{\hat{n} - 1} j |K[j]|^2 \\
  & \leq 2 \, \frac{1}{n} \sum_{j = 1}^{n} j |K[j]|^2
   \xrightarrow[n \to \infty]{} 0,
\end{align*}
 due to the fact that $\sum |K[j]|^2$ is convergent, and using Kronecker's lemma~\cite{ShiryBook}.
 \end{itemize}

Identical arguments hold whenever $n$ is even with the minor modification of replacing $\hat{n}-1$ with $\hat{n} - 2$. Therefore, by Property~\ref{th:asymeq} we conclude that $\Ln \sim \Lnt$. 

Being circulant, the eigenvalues of $\Lnt$ are known to be equal to
\begin{align*}
    \beta_k^{[n]} & = K[0] + \sum_{m = 1}^{\hat{n}-1} K[m] e^{- j 2 \pi \frac{k m}{n}} + \sum_{m = 1}^{\hat{n}-2} K[m] e^{- j 2 \pi \frac{k (n-m)}{n}} \\
    & = \sum_{m = -(\hat{n}-1) + 1}^{\hat{n}-1} K[m] e^{- j 2 \pi \frac{k}{n} m}, \quad k = 0, \cdots (n-1) \\
    & = \sum_{m = - \lfloor n/2 \rfloor + 1}^{\lfloor n/2 \rfloor} K[m] e^{- j 2 \pi \frac{k}{n} m}, \quad k = 0, \cdots (n-1)
\end{align*}
when $n$ is even. Whenever $n$ is odd
\begin{align*}
    \beta_k^{[n]} = \sum_{m = -\lfloor n/2 \rfloor}^{\lfloor n/2 \rfloor} K[m] e^{ - j 2 \pi \frac{k}{n} m}, \quad k = 0, \cdots (n-1).
\end{align*}

If one defines the function
\begin{equation*}
    S^{[n]} \dtf = \left\{ \begin{array}{ll} \displaystyle
    \sum_{m = - \lfloor n/2 \rfloor + 1}^{\lfloor n/2 \rfloor} K[m] e^{- j 2 \pi f m} & n \text{ even} \\ \displaystyle
     \sum_{m = -\lfloor n/2 \rfloor}^{\lfloor n/2 \rfloor} K[m] e^{- j 2 \pi f m} & n \text{ odd}
    \end{array} \right.
\end{equation*}
then one can clearly see that
\begin{align}
    \beta_k^{[n]}  & = S^{[n]} \dtf \bigg|_{f = (k/n)} \nonumber \\
    \& \quad S^{[n]} \dtf & \xrightarrow[n \to \infty]{} S \dtf
    \label{eq:limitSpect}
\end{align}
where $S \dtf$ is the DTFT of the $K[\cdot]$, i.e. the spectrum of $\{ X[\cdot] \}$. Note that the convergence in~\eqref{eq:limitSpect} is uniform due to the fact that $K[\cdot]$ is absolutely summable.

\begin{lemma}
  Let $g(\cdot)$ be a continuous real-valued function. Then
  \begin{equation}
    \lim_{n \to \infty} \frac{1}{n} \sum_{k = 0}^{n-1} g \left( \beta_k^{[n]} \right)
    = \int_0^{1} g \left( S \dtf \right) df.
  \end{equation}
\label{le:avg}
\end{lemma}
\begin{proof}
Let us consider the difference,
\begin{align*}
    & \left| \int_0^{1} g \left( S \dtf \right) df - \frac{1}{n} \sum_{k = 0}^{n-1} g \left( \beta_k^{[n]} \right) \right| \\
    \leq \, & \left| \int_0^{1} g \left( S \dtf \right) df - \frac{1}{n} \sum_{k = 0}^{n-1} g \left( S \left( e^{j 2 \pi (k/n)} \right) \right) \right| \\
    & + \left| \frac{1}{n} \sum_{k = 0}^{n-1} g \left( S \left( e^{j 2 \pi (k/n)} \right) \right) - \frac{1}{n} \sum_{k = 0}^{n-1} g \left( \beta_k^{[n]} \right) \right|.
\end{align*}

By the simple property of Riemann integrals,
\begin{equation*}
    \frac{1}{n} \sum_{k = 0}^{n-1} g \left( S \left( e^{j 2 \pi (k/n)} \right) \right) \xrightarrow[n \to \infty]{} \int_0^{1} g \left( S \dtf \right) df
\end{equation*}
and therefore, for any $\epsilon > 0$, there exists an $n_{0}$ such that for any $n \geq n_{0}$, the first term is less than $\epsilon/2$.

When it comes to the second term, since $S^{[n]} \dtf$ converges uniformly to $S \dtf$ and since $g(\cdot)$ is continuous,  for any $\epsilon > 0$, there exists an $n_{1}$ such that for any $n \geq n_{1}$, 
\begin{align*}
    \left| g \left( S^{[n]} \left( e^{j 2 \pi f} \right) \right) - g \left( S \left( e^{j 2 \pi f} \right) \right) \right| \leq \epsilon/2,
\end{align*}
and therefore, whenever $n \geq n_{1}$,
\begin{multline*}
    \left| \frac{1}{n} \sum_{k = 0}^{n-1} g \left( S \left( e^{j 2 \pi (k/n)} \right) \right) - \frac{1}{n} \sum_{k = 0}^{n-1} g \left( \beta_k^{[n]} \right) \right| \\
    \leq \, \frac{1}{n} \sum_{k = 0}^{n-1} \left| g \left( S \left( e^{j 2 \pi (k/n)} \right) \right) - g \left( \beta_k^{[n]} \right) \right|
    \leq \, \epsilon/2.
\end{multline*}

In conclusion, for any $\epsilon > 0$, there exists an $n_o = \max \{ n_0, n_1 \}$ such that for any $n \geq n_o$,
\begin{align*}
    \left| \int_0^{1} g \left( S \dtf \right) df - \frac{1}{n} \sum_{k = 0}^{n-1} g \left( \beta_k^{[n]} \right) \right| \leq \epsilon.
\end{align*}
\end{proof}

Combining the results of Lemma~\ref{le:avg} and Theorem~\ref{th:contfct} yields the following 
\begin{theorem}
\label{th:CovMatEV}
Let $\{ X[\cdot] \}$ be a DT wide-sense stationary process with absolutely-summable auto-covariance function $K[\cdot]$ and a positive spectrum $S \dtf$ and let $\left\{ \alpha^{[n]}_k \right\}_k$ be the eigenvalues of the auto-covariance matrix of the $(X[m], \cdots, X[m+n])$ for some $m \in \Integers$. 

For any arbitrary continuous function $F(x)$ on $(0,\infty)$,
\begin{equation*}
  \lim_{n \to \infty} \frac{1}{n} \sum_{k=0}^{n-1} F\left(\alpha^{[n]}_{k} \right) = \int_0^1 F \left( S \dtf \right) df,
\end{equation*}
whenever the integral exists.

\end{theorem}

\section{Asymptotic Behavior of Relative Entropy}
\label{appAsymRE}

We establish below the asymptotic behavior of the relative entropy between two Gaussian vectors. It is assumed that these $n$-dimensional vectors are parts of discrete-time (wide-sense)
stationary processes with auto-covariance functions that are square-summable.

Without loss of generality, we assume that the two laws at hand are zero-mean.

\subsection{Relative Entropy between Correlated and Independent Gaussians}
\label{appD}

Assume in what follows that
\begin{equation*}
  \pn = \Normalf{\vx}{\vzero}{\mat{K}} \quad \& \quad \qn = \Normalf{\vx}{\vzero}{\mat{I}},
\end{equation*}
where $\pn$ is the law of an $n$-dimensional "segment" of a wide-sense stationary Gaussian process with absolutely summable auto-covariance function and non-zero spectrum $S \dtf$. The law $\qn$ corresponds to that of an $n$-dimensional "segment" of a wide-sense stationary white Gaussian process with a unit spectrum.
 
Denoting by $\left\{ \alpha^{[n]}_k \right\}_k$ the eigenvalues of $\mat{K}$, it is well-known~\cite{duchi2016DerivationsFL} that the relative entropy is
\begin{align*}
 \KL{\pn}{\qn} & = \frac{1}{2} \text{tr}(\mat{K}) - \frac{1}{2} \log \text{det}(\mat{K}) - \frac{n}{2} \\
  & = \frac{1}{2} \sum_{k=0}^{n-1} \alpha^{[n]}_k - \frac{1}{2} \sum_{k=0}^{n-1} \log \alpha^{[n]}_k  - \frac{n}{2}
\end{align*}
Using Theorem~\ref{th:CovMatEV},
\begin{multline*}
\frac{1}{n} \KL{\pn}{\qn}  = \frac{1}{2} \frac{1}{n} \sum_{k=0}^{n-1} \alpha^{[n]}_k - \frac{1}{2} \frac{1}{n} \sum_{k=0}^{n-1} \log \alpha^{[n]}_k  - \frac{1}{2} \\
 \xrightarrow[n \to \infty]{} \frac{1}{2} \!\int_0^1 \! S \dtf  df - \frac{1}{2} \int_0^1 \! \log S \dtf  df - \frac{1}{2}.
\end{multline*}

Finally, note that if the two hypotheses have identical power, $\int S \dtf = 1$ and the limit of relative entropy simplifies to
\begin{equation*}
\frac{1}{n} \KL{\pn}{\qn} \xrightarrow[n \to \infty]{} \frac{-1}{2} \int_0^1 \log S \dtf \! df,
\end{equation*}
which is by Jensen's strictly positive and equal to zero if and only if $S \dtf = 1$.

\subsection{Relative Entropy between Correlated Gaussians}
\label{appDC}

Now we consider
\begin{equation*}
  \pn = \Normalf{\vx}{\vzero}{\Lambda^{[n]}_p} \quad \& \quad \qn = \Normalf{\vx}{\vzero}{\Lambda^{[n]}_q},
\end{equation*}
where $\pn$ and $\qn$ are the laws of an $n$-dimensional "segment" of two wide-sense stationary Gaussian processes with absolutely summable auto-covariance functions and bounded non-zero spectra $S_p \dtf$ and $S_q \dtf$ respectively.
The relative entropy may be readily derived:
\begin{align*}
& \KL{\pn}{\qn} \\
& \, = \,\frac{1}{2} \text{tr} \left( \Lambda^{[n]}_p \Lambda^{-1 \, [n]}_q \right) - \frac{1}{2} \log \frac{ \text{det}\left(\Lambda^{[n]}_p \right) } {\text{det}\left(\Lambda^{[n]}_q \right) } - \frac{n}{2} \\
\end{align*}

The asymptotics of the last two terms may be readily obtained as in the previous section. It remains to determine those of the first term.

Since $\Lambda^{[n]}_p \sim \tilde{\Lambda}^{[n]}_p$ and $\Lambda^{[n]}_q \sim \tilde{\Lambda}^{[n]}_q$, then
\begin{align}
    \Lambda^{-1 \, [n]}_q & \sim \tilde{\Lambda}^{-1 \, [n]}_q \nonumber \\
    \& \quad \Lambda^{[n]}_p \Lambda^{-1 \, [n]}_q & \sim \tilde{\Lambda}^{[n]}_p \tilde{\Lambda}^{-1 \, [n]}_q
    \label{eq:CovMatProduct}
\end{align}

Since all circulant Hermitian matrices can be diagonalized with the same orthogonal matrix, the eigenvalues of $\tilde{\Lambda}^{[n]}_p \tilde{\Lambda}^{-1 \, [n]}_q$ are simply the ratios of those of $\tilde{\Lambda}^{[n]}_p$ to that of $\tilde{\Lambda}^{[n]}_q$:
\begin{equation*}
    \frac{\beta^{[n]}_k}{\gamma^{[n]}_k} = 
    \frac{S_p^{[n]} \dtf}{S_q^{[n]} \dtf } \Bigg|_{f = (k/n)}
\end{equation*}

On the other hand and as seen in Note~\ref{note:Eig}, the eigenvalues of $\Lambda^{[n]}_p \Lambda^{-1 \, [n]}_q$ are the same as those of the Hermitian $\Lambda^{1/2 [n]}_p \Lambda^{-1 \, [n]}_q \Lambda^{1/2 [n]}_p$ and are hence real and positive. Applying Theorem~\ref{th:contfct}, 
\begin{align*}
& \lim_{n \to \infty} \,  \frac{1}{n} \text{tr} \left( \Lambda^{[n]}_p \Lambda^{-1 \, [n]}_q \right) \\
& \, =  \lim_{n \to \infty} \,  \frac{1}{n} \text{tr} \left( \tilde{\Lambda}^{[n]}_p \tilde{\Lambda}^{-1 \, [n]}_q \right) = \lim_{n \to \infty} \, \frac{1}{n} \sum_{k=0}^{n-1}  \frac{S_p^{[n]} \left( e^{j 2 \pi \frac{k}{n}}\right)}{S_q^{[n]} \left( e^{j 2 \pi \frac{k}{n}}\right) },
\end{align*}
and following the same steps as in the proof of Lemma~\ref{le:avg}, since the spectra $S_p \dtf$ and $S_q \dtf$ are upper-bounded, lower-bounded by a positive scalar and converge uniformly, their ratio
\begin{equation*}
 \frac{S_p^{[n]} \dtf}{S_q^{[n]} \dtf } \xrightarrow[n \to \infty]{} \frac{S_p \dtf}{S_q \dtf } \quad \text{ uniformly,} 
\end{equation*}
and by basic properties of Riemann integrals,
\begin{align*}
\lim_{n \to \infty} \, \frac{1}{n} \sum_{k=0}^{n-1}  \frac{S_p^{[n]} \left( e^{j 2 \pi \frac{k}{n}}\right)}{S_q^{[n]} \left( e^{j 2 \pi \frac{k}{n}}\right) } = \int_0^1 \frac{S_p \dtf}{S_q \dtf } \, df.
\end{align*}

In summary,
\begin{align}
& \lim_{n \to \infty} \, \frac{1}{n} \KL{\pn}{\qn} \nonumber \\
& \, = \frac{1}{2} \, \lim_{n \to \infty} \frac{1}{n} \text{tr} \left( \Lambda^{[n]}_p \Lambda^{-1 \, [n]}_q \right) - \frac{1}{2} \, \lim_{n \to \infty} \frac{1}{n} \log \frac{ \text{det}\left(\Lambda^{[n]}_p \right) } {\text{det}\left(\Lambda^{[n]}_q \right) } - \frac{1}{2} \nonumber \\
& \, = \frac{1}{2} \int_0^1 \frac{S_p \dtf}{S_q \dtf } \, df - \frac{1}{2} \int_0^1 \! \log \frac{S_p \dtf}{S_q \dtf} \, df - \frac{1}{2}.
\label{eq:KLGeneral}
\end{align}
 

\section{Sublinear Growth of $D(\pn||\qn)$}

Consider the following setup (given in~\cite{fahs2024testing}):
Given $n \in \mathbb{N}$, $n \geq 3$, consider the following PDF over $[0,1]^n$:
\begin{align*} 
    \qn (x^{[n]}) = \begin{cases}
    n - \sqrt{n}, & \text{if } \| x^n \|_\infty \leq n^{-1/n}, \vspace{5pt} \\
    \frac{\sqrt{n}}{n-1}, & \text{otherwise},
    \end{cases}
\end{align*}
where $\|x^n \|_\infty = \max_i |x_i|$. It can be easily verified that $\qn$ is a valid PDF as the volume of the set $\{x^{[n]}: \| x^{[n]} \|_\infty \leq n^{-1/n} \}$ is $\frac{1}{n}$. Let $\pn$ be the uniform law over $[0,1]^n$. 

Now, consider the following hypothesis testing problem:
\begin{eqnarray*}
     \left\{ \begin{array}{ll}
        \displaystyle  H_0:  X^n \sim \pn(\cdot), \vspace{5pt} \\ 
        \displaystyle  H_1: X^n \sim \qn(\cdot).
      \end{array} \right.
  \end{eqnarray*}

Consequently,
\begin{align}
    \KL{\pn}{\qn} & = \int \pn \log \frac{\pn}{\qn} dx^{[n]} \\
    & = \frac{1}{n} \log \frac{1}{n-\sqrt{n}} + \left(1-\frac{1}{n}\right) \log \frac{n-1}{\sqrt{n}} \\
    & = \log \frac{n-1}{\sqrt{n}} + \frac{1}{n} \log \frac{\sqrt{n}}{(n-1)(n-\sqrt{n})}, \\
    & \stackrel{\text{(a)}} \sim \log \sqrt{n},
\end{align}
where (a) follows from elementary steps. As such, $\KL{\pn}{\qn}$ grows with $n$ as required in Theorem~\ref{thm:GCSL}.

\begin{lemma}
    Consider the $\delta$-typical set  $\KLTset$. For any constant family $\epsilon$, any constant family $\delta$ is $\epsilon$-good. In particular, the conditions of Theorem~\ref{thm:GCSL} hold and, for any $\tau > 0$,
\begin{align*}
    -\log \beta^{[n]}_\tau = \Dn + o \left( \Dn \right).
\end{align*}
\end{lemma}

\begin{proof}
    Let $B = \{ x^{[n]}: \|x^{[n]} \|_\infty > n^{-1/n} \}$. Then, for $x^{[n]} \in B$, 
    \begin{align}
        \log \frac{\pn}{\qn} dx^{[n]} = \log \frac{n-1}{\sqrt{n}}.
    \end{align}

Hence,
\begin{align}
    \log \frac{\pn}{\qn} dx^{[n]} - \KL{\pn}{\qn} & = \frac{1}{n} \log \frac{\sqrt{n}}{(n-1)(n-\sqrt{n})} \\
    & \xrightarrow[n \to \infty]{} 0.
\end{align}
Therefore, for any constant $\delta$, $B \subseteq \KLTset$. Finally, $\pn(\KLTset) \geq \pn (B)  = 1- 1/n \xrightarrow{n \to \infty} 1.$
\end{proof}
\end{document}